\documentclass[11pt]{article}
\usepackage{fullpage}
\usepackage{latexsym}
\usepackage{amsthm}
\usepackage{amsfonts}
\usepackage{amssymb}
\usepackage[compact]{titlesec}
\usepackage{amsmath}
\usepackage{nicefrac}
\usepackage{epsfig}
\usepackage{ifpdf,color}

\definecolor{Darkblue}{rgb}{0,0,0.4}
\definecolor{Brown}{cmyk}{0,0.81,1.,0.60}
\definecolor{Purple}{cmyk}{0.45,0.86,0,0}
\newcommand{\mydriver}{hypertex}
\ifpdf
 \renewcommand{\mydriver}{pdftex}
\fi
\usepackage[breaklinks,\mydriver]{hyperref}
\hypersetup{colorlinks=true,
            citebordercolor={.6 .6 .6},linkbordercolor={.6 .6 .6},%
citecolor=Darkblue,urlcolor=black,linkcolor=Darkblue,pagecolor=black}
\newcommand{\lref}[2][]{\hyperref[#2]{#1~\ref*{#2}}}

\usepackage{float}

\newtheorem{theorem}{Theorem}[section]
\newtheorem{claim}[theorem]{Claim}
\newtheorem{proposition}[theorem]{Proposition}
\newtheorem{lemma}[theorem]{Lemma}
\newtheorem{corollary}[theorem]{Corollary}
\newtheorem{conjecture}[theorem]{Conjecture}

\theoremstyle{definition}

\newtheorem{definition}[theorem]{Definition}

\def\A{{B}}
\def\E{{\rm E}}
\def\Re{{\mathbb R}}
\def\Pr{{\rm Pr}}

\newcommand{\factor}{\ensuremath{O(\smash{\frac{\log n}{\log\log n}})}}


\newcounter{note}[section]

\newcommand{\initOneLiners}{%
    \setlength{\itemsep}{0pt}
    \setlength{\parsep }{0pt}
    \setlength{\topsep }{0pt}
}
\newenvironment{OneLiners}[1][\ensuremath{\bullet}]
    {\begin{list}
        {#1}
        {\initOneLiners}}
    {\end{list}}

\def\polhk#1{\setbox0=\hbox{#1}{\ooalign{\hidewidth\lower1.5ex\hbox{`}\hidewidth\crcr\unhbox0}}}

\title{An Improved Integrality Gap for Asymmetric TSP Paths\thanks{An
    extended abstract of this paper appears in the \emph{Proceedings of
      the 16th Conference on Integer Programming and Combinatorial
      Optimization, 2013}.}}
\author{
Zachary Friggstad\thanks{Department of Computing Science, University of Alberta.}
\and
Anupam Gupta\thanks{Department of Computer Science, Carnegie Mellon University, Pittsburgh
    PA 15213, and Microsoft Research SVC, Mountain View, CA
    94043. Research was partly supported by NSF awards CCF-0964474 and
    CCF-1016799.}
\and
Mohit Singh\thanks{Microsoft Research, Redmond.}
}

\begin{document}

\maketitle

\begin{abstract}

  \bigskip
  The Asymmetric Traveling Salesperson Path Problem (ATSPP) is one
  where, given an \emph{asymmetric} metric space $(V,d)$ with specified
  vertices $s$ and $t$, the goal is to find an $s$-$t$ path of minimum
  length that passes through all the vertices in $V$.

  \medskip
  This problem is closely related to the Asymmetric TSP (ATSP),
  which seeks to find a tour (instead of an $s$-$t$ path) visiting all
  the nodes: for ATSP, a $\rho$-approximation guarantee implies an
  $O(\rho)$-approximation for ATSPP. However, no such connection is
  known for the \emph{integrality gaps} of the linear programming
  relaxations for these problems: the current-best approximation
  algorithm for ATSPP is $O(\log n/\log\log n)$, whereas the best bound
  on the integrality gap of the natural LP relaxation (the subtour
  elimination LP) for ATSPP is $O(\log n)$.

  \medskip
  In this paper, we close this gap, and improve the current best bound
  on the integrality gap from $O(\log n)$ to $O(\log n/\log\log n)$. The
  resulting algorithm uses the structure of narrow $s$-$t$ cuts in the
  LP solution to construct a (random) tree spanning tree that can be cheaply augmented
  to contain an Eulerian $s$-$t$ walk.

  \medskip
  We also build on a result of Oveis Gharan and Saberi and show
  a strong form of Goddyn's conjecture about thin spanning trees
  implies the integrality gap of the subtour elimination LP relaxation for ATSPP is bounded by a constant.
  Finally, we give a simpler family of instances showing the
  integrality gap of this LP is at least $2$.
\end{abstract}



\section{Introduction}
\label{sec:introduction}

In the Asymmetric Traveling Salesperson Path Problem (ATSPP), we are
given an \emph{asymmetric} metric space $(V,d)$ (i.e., one where the
distances satisfy the triangle inequality, but potentially not the
symmetry condition), and also specified source and sink vertices $s$ and
$t$, and the goal is to find an $s$-$t$ Hamilton path of minimum length.

ATSPP is a close relative of Asymmetric TSP
(ATSP), where the goal is to find a Hamilton tour instead of an $s$-$t$
path.  For ATSP, the $\log_2 n$-approximation of Frieze,
Galbiati, and Maffioli~\cite{FGM} from 1982 was the best result known
for more than two decades, until it was finally improved by constant
factors in~\cite{Blaser08,KLSS05,FS07}. A breakthrough on this problem
was an $\factor$-approximation due to Asadpour, Goemans,
M{\polhk{a}}dry, Oveis Gharan, and Saberi~\cite{AGMSS}; they also
bounded the integrality gap of the subtour elimination linear
programming relaxation for ATSP by the same factor.

Somewhat surprisingly, the study of ATSPP has been of a more recent
vintage: the first approximation algorithms appeared only around
2005~\cite{LN08,CP07,FS07}. It is easily seen that the ATSP reduces to
ATSPP in an approximation-preserving fashion (by guessing two
consecutive nodes on the tour).  In the other direction, Feige and Singh~\cite{FS07}
show that a $\rho$-approximation for ATSP implies an
$O(\rho)$-approximation for ATSPP. Using the above-mentioned
$\factor$-approximation for ATSP~\cite{AGMSS}, this implies an
$\factor$-approximation for ATSPP as well.

The subtour elimination linear program generalizes simply to ATSPP
and is given in \lref[Section]{sec:rounding-algorithm}. However,
prior to our work, the best integrality gap known for this LP for ATSPP was still $O(\log
n)$~\cite{FSS10}. 
In this paper we show the following result.

\begin{theorem}\label{thm:main}
  The integrality gap of the subtour elimination linear program for
  ATSPP is $\factor$.
\end{theorem}

We also explore the connection between integrality gaps for ATSPP
and the so-called ``thin trees conjecture''. In particular, if Goddyn's conjecture
regarding thin trees holds with strong-enough quantitative bounds
then the integrality gap of the subtour elimination LP for ATSPP
is bounded by a constant. The precise statement of the conjecture and of
our result can be found in \lref[Section]{sec:thin}. This is analogous to
a similar statement made by Oveis Gharan and Saberi regarding
the integrality gap of the subtour elimination LP for ATSP~\cite{GS11}.

Finally, we give a simple construction showing that the integrality gap of
this LP is at least~$2$; this example is simpler than previous known
integrality gap instance showing the same lower bound, due to Charikar,
Goemans, and Karloff~\cite{CGK06}.

Given the central nature of linear programs in approximation algorithms,
it is useful to understand the integrality gaps for linear programming
relaxations of optimization problems. Not only does this study give us a
deeper understanding into the underlying problems, but upper bounds
on the integrality gap of LPs are often useful in approximating related problems.
For example, the polylogarithmic approximation guarantees in
the work of Nagarajan and Ravi~\cite{NR07} for Directed Orienteering and
Minimum Ratio Rooted Cycle, and those in the work of Bateni and
Chuzhoy~\cite{BC10} for Directed $k$-Stroll and Directed $k$-Tour were
all improved by a factor of $\log\log n$ following the improved bound of
$\factor$ on the integrality gap of the subtour LP relaxation for ATSP.
We emphasize that these improvements required the integrality gap
bound improvement for ATSP, not merely improved approximation guarantees.

\subsection{Our Approach}
\label{sec:approach}

Our approach to bound the integrality gap for ATSPP is similar to that
for ATSP~\cite{AGMSS,GS11}, but with some crucial differences.
To prove \lref[Theorem]{thm:main}, we sample a
random spanning tree in the underlying undirected multigraph
and then augment the directed version of this tree
to an integral circulation using Hoffman's circulation theorem while
ensuring the $t$-$s$ edge is only used once. The support of this circulation
is weakly connected, so it can be used to obtain an Eulerian circuit with no greater cost.
Deleting the $t$-$s$ edge from this walk results in a spanning $s$-$t$ walk.

However, the non-Eulerian nature of ATSPP makes it difficult
to satisfy the cut requirements in Hoffman's circulation theorem if we
sample the spanning tree directly from the distribution given by the LP
solution. It turns out that the problems come from the $s$-$t$ cuts $U$
that are nearly-tight: i.e., which satisfy $1 < x^*(\partial^+(U)) <
1+\tau$ for some small constant $\tau$ --- these give rise to problems
when the sampled spanning tree includes more than one edge across this
cut. Such problems also arise in the symmetric TSP paths case (studied
in the recent papers of An, Kleinberg, and Shmoys~\cite{AKS12} and Seb\H{o}~\cite{Sebo13}): their
approach is again to take a random tree directly from the distribution
given by the optimal LP solution, but in some cases they need to boost
the narrow cuts, and they show that the loss due to this boosting is
small.

In our case, the asymmetry in the problem means that boosting the narrow
cuts might be prohibitively expensive. Hence, our idea is to preprocess
the distribution given by the LP solution to \emph{tighten} the narrow
cuts, so that we never pick two edges from a narrow cut. Since the
original LP solution lies in the spanning tree polytope, lowering the
fractional value on some edges means we need to raise the fractional value on other edges. This would cause
the costs to increase, and the technical heart of the paper is to ensure
this can be done with a small increase in the cost.

Our approach for proving an $O(1)$ integrality gap bound
under the thin trees conjecture is similarly inspired by related work
for ATSP~\cite{GS11},
but, again, we must be careful to ensure that the thin tree
crosses each narrow cut exactly once.
We do this by finding a cheap thin tree ``between'' narrow cuts
(which we will prove are nested) and then chaining these thin together
trees by selecting a single edge across each narrow cut.
The resulting tree will have the desired thinness properties.




\subsection{Other Related Work}
\label{sec:related-work}

The first non-trivial approximation for ATSPP was an
$O(\sqrt{n})$-approximation by Lam and Newman~\cite{LN08}. This was
improved to $O(\log n)$ by Chekuri and P\'al~\cite{CP07}, and the
constant was further improved in~\cite{FS07}. The paper~\cite{FS07} also
showed that a $\rho$-approximation algorithm for ATSP can be used
to obtain an $O(\rho)$-approximation algorithm for ATSPP.
All these results are combinatorial and do not bound
integrality gap of ATSPP. A bound of $O(\sqrt{n})$ on the integrality
gap of ATSPP was given by Nagarajan and Ravi~\cite{NR-direct-latency},
and was improved to $O(\log n)$ by Friggstad, Salavatipour and
Svitkina~\cite{FSS10}. Note that there is still no result known that
relates the integrality gaps of subtour elimination relaxations for ATSP and ATSPP in a
black-box fashion.

In the symmetric case (where the problems become TSPP and TSP
respectively), constant factor approximations and integrality gaps have
long been known. We do not survey the rich body of literature on TSP
here, instead pointing the reader to, e.g., the recent paper on
graphical TSP by Seb\H{o} and Vygen~\cite{SV12}. An exception is a
result of An, Kleinberg, and Shmoys~\cite{AKS12}, who give an upper
bound of $1.618$ on integrality gap of the LP relaxation for the TSPP
problem; their algorithm also proceeds via studying the narrow $s$-$t$
cuts, and the connections to our work are discussed in
\lref[Section]{sec:approach}.  This bound on the integrality gap was
subsequently improved to $1.6$ via a more refined analysis by
Seb\H{o}~\cite{Sebo13}.

\subsection{Notation and Preliminaries}

Given a directed graph $G = (V,A)$, and two disjoint sets $U, U'
\subseteq V$, let $\partial(U;U') = A \cap (U \times U')$. We use the
standard shorthand that $\partial^+(U) := \partial(U; V \setminus U)$,
and $\partial^-(U) := \partial(V \setminus U; U)$. When the set $U$ is a
singleton (say $U = \{u\}$), we use $\partial^+(u)$ or $\partial^-(u)$
instead of $\partial^+(\{u\})$ or $\partial^-(\{u\})$. For undirected
graph $H = (V,E)$, we use $\partial(U; U')$ to denote edges crossing
between $U$ and $U'$, and $\partial(U)$ to denote the edges with exactly
one endpoint in $U$ (which is the same as $\partial(V \setminus U))$.
For any subset $U \subseteq V$ we let $A(U)$ denote $A \cap (U \times U)$, the set of
arcs with both endpoints in $U$.
If we are discussing subsets of arcs $\A$ of $G$, we add subscripts
to the $\partial$ notation to indicate we only consider those arcs crossing
the cut that in are $\A$.
For example, $\partial_{\A}^+(U)$ denotes $\partial^+(U) \cap  \A$
and so on. A collection of subsets of $V$, say $\pi$ is a partition if each element of $V$ occurs in exactly one part of $\pi$. Given a graph $G=(V,E)$ and a partition
$\Pi$ of $V$, we let $\partial(\pi)$ to be the set of edges in $E$ which have endpoints in different sets of $\pi$.

For a digraph $G = (V,A)$, a set of arcs $B \subseteq A$ is \emph{weakly
  connected} if the undirected version of $B$ forms a connected graph
that spans all vertices in $V$.

For values $x_a \in \Re$ for all $a \in A$, and a set of arcs $B
\subseteq A$, we let $x(B)$ denote the sum $\sum_{a \in B} x_a$.

Given an undirected graph $H = (V,E)$ and a subset of edges
$F \subseteq E$, we let $\chi_F \in \{0,1\}^{|E|}$ denote
the characteristic vector $F$. The spanning
tree polytope is the convex hull of $\{ \chi_T \mid T \text{ spanning
  tree of } H \}$. See, e.g.,~\cite[Chapter~50]{Schrijver-book} for
several equivalent linear programming formulations of this polytope.
We sometimes abuse notation and call a set of directed arcs $T$
a tree if the undirected version of $T$ is a tree in the usual sense.

A \emph{directed metric graph} on vertices $V$ has arcs $A = \{uv : u,v \in V, u \neq v\}$
where the non-negative arc costs satisfy the triangle
inequality $c_{uv} + c_{vw} \geq c_{uw}$ for all $u,v,w \in V$. However,
arcs $uv$ and $vu$ need not have the same cost. An instance of ATSPP
is a directed metric graph along with distinguished vertices $s
\neq t$.


\section{The Rounding Algorithm}
\label{sec:rounding-algorithm}

In this section, we give the linear programming relaxation for
ATSPP, and show how to round a feasible solution $x$ to this LP to get a path of
cost $\factor$ times the cost of $x$.
We then give the proof, with some of the details being deferred to the
following sections.

Given a directed metric graph $G = (V,A)$ with arc costs $\{c_a\}_{a \in
  A}$, we use the following standard linear programming relaxation for
ATSPP which is also known as the subtour elimination linear program.
  \begin{align}
    {\rm minimize}: \sum_{a \in A} c_a x_a & & \tag{{\em
        ATSPP}} \label{eq:lp} \\
    {\rm s.t.}:~~
     x(\partial^+(s)) = x(\partial^-(t)) & =  1 & \\
     x(\partial^-(s)) = x(\partial^+(t)) & =  0 & \\
     x(\partial^+(v)) = x(\partial^-(v)) & =  1 & \forall~ v \in V\setminus \{s,t\}\\
     x(\partial^+(U)) & \geq  1 & \forall~ \{s\} \subseteq U \subsetneq V \\
     x_a & \geq  0 & \forall~ a \in A \notag
\end{align}
%

Constraints (4) can be separated over in polynomial time using
standard min-cut algorithms, so this LP can be solved
in polynomial time using the ellipsoid method.
We begin by solving the above LP to obtain an optimal solution $x^*$.
Consider the undirected (multi)graph $H = (V,E)$ obtained by removing
the orientation of the arcs of $G$. That is, create precisely two edges
between every two nodes $u,v \in V$ in $H$, one having cost $c_{uv}$ and
the other having cost $c_{vu}$. (Hence, $|E| = |A|$.) For a point $w \in
\Re_+^A$, let $\kappa(w)$ denote the corresponding point in $\Re_+^E$,
and view $\kappa(w)$ as the ``undirected'' version of $w$.

%

We will use the following definition: An \emph{$s$-$t$ cut} is a subset
$U \subset V$ such that $\{s\} \subseteq U \subseteq V\setminus
\{t\}$. The following fact will be used throughout the paper.

\begin{claim} \label{claim:lp}
Let $x^*$ be a feasible solution to LP (\ref{eq:lp}). For any $s-t$ cut $U$, $x^*(\partial^+(U)) - x^*(\partial^-(U)) = 1$.
Also, $x(\partial^+(U)) = x^*(\partial^-(U))$ for every nonempty $U \subseteq V \setminus \{s,t\}$.
\end{claim}
\begin{proof}
For any nonempty subset of vertices $U$ we have
\begin{eqnarray*}
x^*(\partial^+(U)) - x^*(\partial^-(U)) & = & \left(\sum_{e \in \partial^+(U)} x^*_e - \sum_{e \in A(U)}  x^*_e\right) - \left(\sum_{e \in \partial^-(U)} x^*_e - \sum_{e \in A(U)} x^*_e \right) \\
& = & \sum_{v \in U} x(\partial^+(v)) - \sum_{v \in U} x(\partial^-(v)). \\
\end{eqnarray*}
If $U$ is an $s-t$ cut, then the first sum in the last expression is $|U|$ and the second sum is $|U|-1$ by Constraints (1), (2), and (3). If $U \subseteq V \setminus \{s,t\}$,
then both sums are equal to $|U|$ by Constraints (3).
\end{proof}

\begin{definition}[Narrow cuts]
  Let $\tau \geq 0$. An $s$-$t$ cut $U$ is \emph{$\tau$-narrow} if
  $x^*(\partial^+(U)) < 1 + \tau$ (or equivalently, $x^*(\partial^-(U)) <
  \tau$).
\end{definition}

The main technical lemma is the following:
\begin{lemma}\label{lem:fix_sol}
  For any $\tau \in [0,1/4]$, one can find, in polynomial-time, a vector
  $z \in [0,1]^A$ (over the directed arcs) such that:
  \begin{enumerate}
  \item[(a)] its undirected version $\kappa(z)$ lies in the spanning tree
    polytope for $H$,
  \item[(b)] $z \leq \frac{1}{1-3\tau} \, x^*$ (where the inequality
    denotes component-wise dominance), and
  \item[(c)] $z(\partial^+(U)) = 1$ and $z(\partial^-(U)) = 0$ for every
    $\tau$-narrow $s$-$t$ cut $U$.
\end{enumerate}
\end{lemma}

Before we prove the lemma (in \lref[Section]{sec:proof-of-lem}), let us
sketch how it will be useful to get a cheap ATSPP solution. Since $z$ (or more
correctly, its undirected version $\kappa(z)$) lies in the spanning tree
polytope, it can be represented as a convex combination of spanning
trees.  Using some recently-developed algorithms (e.g., those due
to~\cite{AGMSS,CVZ10}) one can choose a (random) spanning tree that crosses
each cut only $\factor$ times more than the LP solution. Finally, we can
use $\factor$ times the LP solution to patch this tree to get an $s$-$t$
path. Since the LP solution is ``weak'' on the narrow cuts and may
contribute very little to this patching (at most $\tau$), it is crucial
that by property~(c) above, this tree will cross the narrow cuts
\emph{only once}, and that too, it crosses in the ``right'' direction,
so we never need to use the LP when verifying the cut conditions of
Hoffman's circulation theorem on narrow cuts. The details of these
operations appear in \lref[Section]{sec:patching}.

We will assume $n \geq 7$ to ensure all of our arguments work. For $n \leq 6$,
we use the known integrality gap bound of $2 \lfloor \log_2 n \rfloor + 1 \leq 5$
from~\cite{FSS10} to ensure the gap is bounded for all $n \geq 2$.

\subsection{The Structure of Narrow Cuts}
\label{sec:proof-of-lem}

We now prove \lref[Lemma]{lem:fix_sol}: it says that we can take the LP
solution $x^*$ and find another vector $z$ such that if an $s$-$t$ cut is
narrow in $x^*$ (i.e. $x^*(\partial^+(U)) < 1 + \tau$), then $z(\partial^+(U)) = 1$. Moreover,
the undirected version of $z$ can be written as a convex combination of
spanning trees, and $z_a$ is not much larger than $x^*_a$ for any arc $a$.

The undirected version of $x^*$ itself can be written as a
convex combination of spanning trees, so if we force $z$ to cross the
narrow cuts to an extent less than $x^*$ (loosely, this reduces the
connectivity), we had better increase the value on other arcs. To show we
can perform this operation without changing any of the coordinates by
very much, we need to study the structure of narrow cuts more closely.
(Such a study is done in the \emph{symmetric} TSP path paper of An et
al.~\cite{AKS12}, but our goals and theorems are somewhat different.)

First, say two $s$-$t$ cuts $U$ and $W$ \emph{cross} if $U\setminus W$
and $W \setminus U$ are
non-empty. 

\begin{lemma}
  \label{lem:nocross}
  For $\tau \leq 1/4$, no two $\tau$-narrow $s$-$t$ cuts cross.
\end{lemma}
\begin{proof}
  Suppose $U$ and $W$ are crossing $\tau$-narrow $s$-$t$ cuts. Then
  \begin{eqnarray*}
    2 + 2\tau & > & x^*(\partial^+(U)) + x^*(\partial^+(W)) \\
    & = & x^*(\partial^+(U \setminus W)) + x^*(\partial^+(W \setminus U)) + x^*(\partial^+(U \cap W)) \\ & &+ x^*(\partial(U \cap W; V \setminus (U \cup W))) - x^*(\partial((U \cup W) \setminus (U \cap W) ; U \cap W)) \\
    & \geq & 1 + 1 + 1 + 0 - 2\tau \\
    & = & 3 - 2\tau
  \end{eqnarray*}
  where the last inequality follows from the first three terms being
  cuts excluding $t$ and hence having at least unit $x^*$-value crossing them (by the
  LP constraints), the fourth term being non-negative, and the last term
  being the $x^*$-value of a subset of the arcs in $\partial^-(U) \cup
  \partial^-(W)$ and remembering that $U$ and $W$ are $\tau$-narrow.
  However, this contradicts $\tau \leq 1/4$.
\end{proof}

\lref[Lemma]{lem:nocross} says that the $\tau$-narrow cuts form a chain
$\{s\} = U_1 \subset U_2 \subset \ldots \subset U_k = V\setminus\{t\}$ with $k \geq 2$.
For $1 < i \leq k$. let $L_i := U_i \setminus U_{i-1}$. We also define
$L_1 = \{s\}$ and $L_{k+1}=\{t\}$. Let $L_{\leq i} := \bigcup_{j = 1}^i L_i$ and $L_{\geq
  i} := \bigcup_{j=i}^{k+1} L_i$.
For the rest of this paper, we will use $\tau$ to denote a value in the range $[0, 1/4]$.
Ultimately, we will set $\tau := 1/4$ for the final bound but we state the lemmas
in their full generality for $\tau \leq 1/4$.

Next, we show that out of the (at most) $1+\tau$ mass of $x^*$ across
each $\tau$-narrow cut $U_i$, most of it comes from the ``local'' arcs
in $\partial(L_i; L_{i+1})$.

\begin{lemma} \label{lem:li+}
  For each $1 \leq i \leq k$; $x^*(\partial(L_i; L_{i+1})) \geq
  1-3\tau$. 
\end{lemma}
\begin{proof}
  If $k = 1$ then $\{s\} = U_1 = U_k = V \setminus \{t\}$ so in fact $V = \{s,t\}$. In this case, $L_1 = \{s\}$
  and $L_2 = \{t\}$ and the LP constraints clearly imply $\partial(L_1;L_2) = 1$.

  Now consider the case $k \geq 2$.
  For $i=1$, since $s,t \not\in L_2$ we have $x^*(\partial^-(L_2)) \geq 1$ from the LP
  constraints. We also have $x^*(\partial^-(U_2)) < \tau$ because $U_2$ is
  $\tau$-narrow, and therefore $x^*(\partial(L_1; L_2))\geq 1-\tau$. A similar argument for $i = k$ shows
  $x^*(\partial(L_k; L_{k+1})) \geq 1-\tau$. So it remains to consider
  $1 < i < k$.
  Define the following quantities, some of which can be zero.
  \begin{OneLiners}
  \item $A = x^*(\partial(L_i; L_{i+1}))$
  \item $B = x^*(\partial(L_i; L_{\geq i+2}))$
  \item $C = x^*(\partial(L_{\leq i-1}; L_{i+1}))$
  \end{OneLiners}
  We have
  \[ 1 \leq x^*(\partial^+(L_i)) = A+B+x^*(\partial(L_i; L_{\leq i-1})) \leq
  A+B + \tau,
  \] because $\partial(L_i; L_{\leq i-1}) \subseteq
  \partial^-(U_{i-1})$ and $U_{i-1}$ is $\tau$-narrow.  Similarly
  \[ 1 \leq x^*(\partial^-(L_{i+1})) = A+C+x^*(\partial(L_{\geq i+2};
  L_{i+1})) \leq A+C + \tau.
  \]
  Summing these two inequalities yields $2 \leq A + (A + B + C) + 2\tau
  \leq A + (1+\tau) + 2\tau$ where we have used $A+B+C \leq
  x^*(\partial^+(U_i)) \leq 1 + \tau$. Rearranging shows $A \geq 1-3\tau$.
\end{proof}



Now, recall that $\kappa(x^*)$ denotes the assignment of arc weights to
the graph $H = (V,E)$ from the previous section obtained by ``removing''
the directions from arcs in $A$.  We prove that the restriction of
$\kappa(x^*)$ to any $L_i$ almost satisfies the partition inequalities
that characterize the convex hull of connected spanning subgraphs of $H$. This characterization was given by Edmonds~\cite{Edmonds70b}; see also Chapter 50, Corollary 50.8(a) in Schrijver~\cite{Schrijver-book}. We state it here for completeness.

\begin{theorem}\cite{Edmonds70b}
Let $G=(V,E)$ be a graph. Then the convex hull of all connected spanning subgraphs of $G$ is given by ${\mathcal C}(G)=\{x\in \Re^E: x(\partial(\pi))\geq |\pi|-1 \;\;\; \forall \textrm{ partitions } \pi \textrm{ of  } V, \;\; 0\leq x\leq 1\}$. Moreover, the convex hull of spanning trees of $G$ is given by
${\mathcal C}(G)\cap \{x \in \Re^E : \sum_{e\in E} x_e= |V|-1\}$.
\end{theorem}


For a partition $\pi = \{W_1, \ldots, W_\ell\}$ of a subset of $V$, we let
$\partial(\pi)$ denote the set of edges whose endpoints lie in two different
sets in the partition. To be clear, $\partial(\pi)$ does not contain any edge that has at least one endpoint in $V-\cup_{i=1}^\ell W_i$.

\begin{lemma}\label{lem:parts}
For any $1 \leq i \leq k+1$ and any partition $\pi = \{W_1, \ldots, W_\ell\}$ of $L_i$, we have $\kappa(x^*)(\partial(\pi))\geq \ell-1-2\tau$.
\end{lemma}
\begin{proof}
Since $L_1 = \{s\}$ and $L_{k+1} = \{t\}$, there is nothing to prove for $i = 1$ or $i = k+1$. So, we suppose $1 < i < k+1$.

Consider the quantity $\alpha = \sum_{j=1}^\ell x^*(\partial^+(W_j)) + x^*(\partial^-(W_j))$.
On one hand $x^*(\partial^+(W_j)) = x^*(\partial^-(W_j)) \geq 1$ because neither $s$ nor $t$ lie in $W_j$ for any $1 \leq j \leq \ell$, so
$\alpha \geq 2\ell$. On the other hand, $\alpha$ counts each arc between two parts in $\pi$ exactly twice and each
arc with one end in $L_i$ and the other not in $L_i$ precisely once. So, $\alpha = 2\kappa(x^*)(\partial(\pi)) + x^*(\partial^+(L_i)) + x^*(\partial^-(L_i))$.

Notice that $\partial^+(L_i)$ and $\partial^-(L_i)$ are disjoint subsets of $\partial^+(U_{i-1}) \cup \partial^-(U_{i-1}) \cup \partial^+(U_i) \cup \partial^-(U_i)$.
So, since both $U_{i-1}$ and $U_i$ are $\tau$-narrow, $x(\partial^+(L_i)) + x(\partial^-(L_i)) < 2+4\tau$.
This shows $2\ell \leq \alpha < 2\kappa(x^*)(\partial(\pi)) + 2 + 4\tau$ which, after rearranging, is what we wanted to show.
\end{proof}

\begin{corollary}\label{cor:parts}
  For any partition $\pi$ of $L_i$, we have
  $\frac{\kappa(x^*)(\partial(\pi))}{1-2\tau}\geq |\pi|-1$.
\end{corollary}
\begin{proof}
From \lref[Lemma]{lem:parts}, we have $\frac{\kappa(x^*)(\partial(\pi))}{1-2\tau} \geq \frac{|\pi|-1-2\tau}{1-2\tau} \geq |\pi|-1$ for any $|\pi| \geq 2$.
\end{proof}


Finally, to efficiently implement the arguments in the proof of \lref[Lemma]{lem:fix_sol}, we need to be able to efficiently find
all $\tau$-narrow cuts $U_i$.
This is done by a standard recursive algorithm that exploits the fact that the cuts are nested.
\begin{lemma} \label{lem:efficient}
There is a polynomial-time algorithm to find all $\tau$-narrow $s-t$ cuts.
\end{lemma}
\begin{proof}
Consider following recursive algorithm.
As input, the routine is given a directed graph $H = (V',A')$ with arc weights $x^*_a$
and distinct nodes $s',t'$  where both $\{s'\}$ and $V'\setminus \{t'\}$ are $\tau$-narrow.
Say a $\tau$-narrow cut $U$ in $H$ is non-trivial if $U \neq \{s'\}$ and $U \neq V' \setminus \{t'\}$.
The claim is that the procedure will find all non-trivial $\tau$-narrow $s-t$ cuts of $H$, provided that they are nested.

The procedure works as follows.
If there are non-trivial $\tau$-narrow $s-t$ cuts in $H$,
then there are nodes $u,v \in V' \setminus\{s',t'\}$ such that
some $\tau$-narrow $s'-t'$ cut $U$ has $\{s',u\} \subseteq U \subseteq V' \setminus \{t',v\}$.
So, the procedure tries all $O(|V'|^2)$ pairs of distinct nodes $u,v$, contracts both $\{s',u\}$ and $\{t',v\}$
to a single node and determines if the minimum cut separating these contracted nodes
has $x^*$-capacity less than $1+\tau$. If such a cut $U$ was found for some $u,v$, the algorithm makes two recursive calls,
one with the contracted graph $H[V'/U]$ with start node being the contraction of $U$ and end node being $t'$,
and the other with the contracted graph $H[V'/(V' \setminus U)]$ with start node $s'$ and end node being the contraction of $V'\setminus U$.
After both recursive calls complete, the algorithm returns all $\tau$-narrow
cuts found by these two recursive calls (of course, after expanding the contracted nodes)
and the $\tau$-narrow cut $U$ itself. If such a cut $U$ was not found over all choices of $u,v$, then the algorithm returns
nothing because there are no non-trivial $\tau$-narrow $s'-t'$ cuts in $H$.

It is easy to see that a non-trivial $\tau$-narrow cut in either contracted graph corresponds to a $\tau$-narrow cut in $H$.
On the other hand, if the $\tau$-narrow $s'-t'$ cuts are nested in $H$, then every non-trivial $\tau$-narrow $s'-t'$
cut apart from $U$ itself corresponds to a non-trivial $\tau$-narrow cut in exactly one of $H[V'/U]$ or
$H[V'/(V'\setminus U)]$. Also, the $\tau$-narrow cuts in both contracted graphs remain nested.
So, the recursive procedure finds all non-trivial $\tau$-narrow cuts of $H$. The number of recursive calls
is at most the number of non-trivial $\tau$-narrow cuts, and this is at most $|V'|$ because the cuts are nested
so it is an efficient algorithm. We call this algorithm initially with graph $G$, start node $s$ and end node $t$.
\lref[Lemma]{lem:nocross} implies the $\tau$-narrow $s-t$ cuts of $G$ are nested
so the recursive procedure finds all non-trivial $\tau$-narrow cuts of $G$. Adding these to
$\{s\}$ and $V\setminus\{t\}$ gives all $\tau$-narrow cuts of $G$.
\end{proof}

\begin{proof}[Proof of Lemma \ref{lem:fix_sol}]
  The claimed vector $z$ can be described by linear constraints: indeed,
  consider the following polytope on the variables $z$.
  \begin{align}
    \kappa(z)(\partial(\pi)) & \geq  |\pi|-1 &\forall {\rm ~partitions~}
    \pi {\rm ~of~ }V \label{cons1}  \\
    \sum_a z_a & = n-1 & \label{cons1-1} \\
    z_a & \leq  {\textstyle \frac{1}{1-3\tau}} \, x^*_a & \forall~ a \in
    A \label{cons2}\\
    z(\partial^+(U_i)) & =  1 & \forall~ \tau\text{-narrow~}\text{$s$-$t$}{\rm
      ~cuts~} U_i \label{cons3}\\
    z(\partial^-(U_i)) & =  0 & \forall~ \tau\text{-narrow~}\text{$s$-$t$}{\rm ~cuts~}
    U_i\label{cons4}\\
    z_a&\geq 0  &\forall~a\in A \label{cons5}
\end{align}

Consider the vector $z$ given as follows.
\begin{equation}
\label{eq:7}
z_{a} = \begin{cases}
\frac{x_a^*}{x^*(\partial(L_i;L_{i+1}))} & \text{if $a \in \partial(L_i;L_{i+1})$ for some $i$;}\\
\frac{x_a^*}{1-2\tau} & \text{if $a\in A(L_i)$ for some $i$;}\\
0 & \text{otherwise.}
\end{cases}
\end{equation}

Constraints~(\ref{cons4}) and (\ref{cons5}) are satisfied by
construction. Constraint~(\ref{cons2}) follows from \lref[Lemma]{lem:li+}
for edges in $\partial(L_i;L_{i+1})$ and by construction for rest of the
edges. For constraint~(\ref{cons3}), note that
\[
z(\partial^+(U_i))=z(\partial(L_i;L_{i+1}))+z(\partial^+(U_i)\setminus
\partial(L_i;L_{i+1}))=\frac{x^*(\partial(L_i;L_{i+1}))}{x^*(\partial(L_i;L_{i+1}))}+0=1.
\]

Next we show Constraints~(\ref{cons1}) holds.  It suffices to
show that $\kappa(z)$ can be decomposed as a convex combination of characteristic
vectors of connected graphs. For
$1 \leq i \leq k+1$, let $z^i$ denote the restriction of $\kappa(z)$ to
edges whose endpoints are both contained in $L_i$. Then
\lref[Corollary]{cor:parts}, Constraints~(\ref{cons5}), and
\cite[Corollary~50.8a]{Schrijver-book} imply that $z^i$ can be
decomposed as a convex combination of integral vectors, each of which
corresponds to an edge set that is connected on $L_i$.  Next, let $z'$
denote the restriction of $\kappa(z)$ to edges whose endpoints are both
contained in some common $L_i$. Since the sets $E(L_1), \ldots,
E(L_{k+1})$ are disjoint, we have that $z' = \sum_i z^i$ (where the
addition is component-wise). Furthermore,
$z'$, being the sum of the $z^i$ vectors,
can be decomposed as a convex combination of integral vectors
corresponding to edge sets $E'$ such that the connected components of
the graph $H' = (V, E')$ are precisely the sets
$\{L_i\}_{i=1}^{k+1}$.

Next, let $z''$ denote the restriction of $\kappa(z)$ to edges contained
in one such $\partial(L_i; L_{i+1})$.  We also note that the sets
$\partial(L_1; L_2), \ldots, \partial(L_k; L_{k+1})$ are disjoint.  By
construction, we have $z''(\partial(L_i; L_{i+1})) = 1$ for each $1 \leq
i \leq k$ so we may decompose $z''$ as a convex-combination of integral
vectors, each of which includes precisely one edge across each
$\partial(L_i; L_{i+1})$.

Adding any integral point $y'$ in the decomposition of $z'$ to any
integral point $y''$ in the decomposition of $z''$ results
in an integral vector that corresponds to a connected graph:
each $L_i$ is connected by $y'$ and consecutive $L_i$ are connected
by $y''$. By construction of $z$, we have $\kappa(z) = z' + z''$
so we may write $z$ as a convex combination of characteristic vectors
of connected graphs, each of which satisfies Constraints~(\ref{cons1}).

Finally, we modify $z$ slightly to ensure constraint (\ref{cons1-1}) holds while maintaining
the other constraints. From~\cite[Corollary~50.8a]{Schrijver-book} and
Constraints (\ref{cons1}) and (\ref{cons5}),
$\kappa(z)$ lies in the convex hull of incidence vectors corresponding to connected
(multi)graphs. Decompose $\kappa(z)$ into a convex combination of such vectors,
drop edges from the corresponding connected graphs to get spanning tree,
and recombine these spanning trees to get a point in the spanning
tree polytope. Note that we only decreased $z_a$ values
so Constraints (\ref{cons2}) and (\ref{cons4}) continue to hold.
Finally, since $\kappa(z)$ now lies in the spanning tree polytope
then each tree must still cross each narrow cut, so Constraints (\ref{cons3}) still hold.

Such a vector can be found efficiently because Constraints (\ref{cons1}) admit
an efficient separation oracle~\cite[Corollary~51.3b]{Schrijver-book}.
\end{proof}

\section{Obtaining an $s$-$t$ Path}
\label{sec:patching}

Having transformed the optimal LP solution $x^*$ into the new vector $z$
(as in \lref[Lemma]{lem:fix_sol}) without increasing it too much in any
coordinate, we now sample a random tree such that it has a small total
cost, and that the tree does not cross any cut much more than prescribed
by $x^*$. Finally we add some arcs to this tree (without increasing its
cost much) so that every $v \not\in \{s,t\}$ has equal indegree and outdegree
while ensuring that $s$ has outdegree 1 and indegree 0.
This gives us an Eulerian $s$-$t$ walk.

By the triangle inequality,
shortcutting this walk past repeated nodes yields a Hamiltonian $s-t$ path
of no greater cost. While this general approach is similar
to that used in~\cite{AGMSS}, some new ideas are required because we are
working with the LP for ATSPP---in particular, only one unit of flow is
guaranteed to cross $s$-$t$ cuts, which is why we needed to deal with
narrow cuts in the first place. The details appear in the rest of this section.

\subsection{Sampling a Tree}\label{sec:sample}

For a digraph $G = (V,A)$ and a collection of arcs $\A \subseteq A$, we say $\A$ is
\emph{$\alpha$-thin with respect to $x^*$} if $|\A
\cap \partial^+(U)| \leq \alpha x^*(\partial^+(U))$ for every $\emptyset
\subsetneq U \subsetneq V$. The set $\A$ is also
\emph{$\beta$-approximate with respect to $x^*$} if the total cost of
all arcs in $\A$ is at most $\beta$ times the cost of $x^*$, i.e.,
$\sum_{a \in \A} c_a \leq \beta \sum_{a \in A} c_a x_a^*$.  The reason
we are deviating from the undirected setting used in~\cite{AGMSS} to the directed setting is that the
orientation of the arcs across each $\tau$-narrow cut will be important
when we sample a random ``tree''.

\begin{lemma} \label{lem:findtree}
  Let $\tau \in [0, 1/4]$. Let $\beta = \frac{3}{1-3\tau}$ and $\alpha =
  \left(2 + \frac{1}{\tau}\right)\cdot \frac{24 \log n}{\log\log n}$.
  For $n \geq 7$, there is a randomized,
  polynomial time algorithm that, with probability at least 1/2, finds
  an $\alpha$-thin and $\beta$-approximate (with respect to $x^*$)
  collection of arcs $\A$ that is weakly connected and satisfies $|\A
  \cap (\partial^+(U))| = 1$ and $|\A \cap (\partial^-(U))| = 0$ for each
  $\tau$-narrow $s$-$t$ cut $U$.
\end{lemma}
\begin{proof}
  Let $z$ be a vector as promised by \lref[Lemma]{lem:fix_sol}. From
  $\kappa(z)$, randomly sample a set of arcs $\A$ whose undirected
  version $T$ is a spanning tree on $V$. This should be done
  from any distribution with the following two properties:
  \begin{itemize}
  \item[(i)] \emph{(Correct Marginals)} $\Pr[e \in T] = \kappa(z)_e$
  \item[(ii)] \emph{(Negative Correlation)} For any subset of edges $F \subseteq E$, $\Pr[F \subseteq
    T] \leq \prod_{e \in F} \Pr[e \in T]$
  \end{itemize}
  This can be obtained using, for example, the swap rounding approach
  for the spanning tree polytope given by Chekuri et al.~\cite{CVZ10}.
  As in~\cite{AGMSS}, the negative correlation property implies the following
  theorem. The proof is found in \lref[Section]{sec:chernoff}.
  \begin{theorem}\label{thm:chernoff}
  For $n \geq 7$, the tree $T$ is $\alpha$-thin with probability at least $1 - \frac{1}{n-1}$.
  \end{theorem}

  By \lref[Lemma]{lem:fix_sol}(b), property~(i) of the random sampling,
  and Markov's inequality, we get that $\A$ (from \lref[Lemma]{lem:findtree}) is
  $\frac{3}{1-3\tau}$-approximate with respect to $x^*$ with probability
  at least 2/3. By a trivial union bound, for $n \geq 7$ we have
  with probability at least 1/2 that $\A$ is both $\alpha$-thin and $\beta$-approximate with
  respect to $x^*$. It is also weakly connected---i.e., the undirected
  version of $\A$ (namely, $T$) connects all vertices in $V$.

  The statement for $\tau$-narrow $s$-$t$ cuts follows from the fact
  that $z$ satisfies \lref[Lemma]{lem:fix_sol}(c). That is, $\A$ contains
  no arcs of $\partial^-(U)$, since $z(\partial^-(U)) = 0$ (for $U$ being a
  $\tau$-narrow $s$-$t$ cut). But since $T$ is a spanning tree,
  $\A$ must contain at least one arc from $\partial^+(U)$. Finally, since
  $z(\partial^+(U))$ is exactly 1, then any set of arcs supported by this
  distribution we use must have precisely one arc from $\partial^+(U)$.
\end{proof}

\subsection{Augmenting to an Eulerian $s$-$t$ Walk}

We wrap up by augmenting the set of arcs $\A$ to an Eulerian $s$-$t$ walk.
Specifically, we prove the following for general $\alpha \geq 1$.
\begin{theorem}\label{thm:augment}
Suppose we are given a collection of arcs $\A$ that is weakly connected, $\alpha$-thin, and satisfies
$|\partial_{\A}^+(U)| = 1$ and $|\partial_{\A}^-(U)| = 0$ for any $\tau$-narrow $s-t$ cut $U$.
We can find a Hamiltonian $s-t$ path with cost at most
$c(\A) +  (1+\tau^{-1}) \alpha  \sum_{a \in A} c_a x^*_a$ in polynomial time.
\end{theorem}

For this, we use Hoffman's circulation theorem, as
in~\cite{AGMSS}, which we recall here for convenience (see,
e.g,~\cite[Theorem~11.2]{Schrijver-book}):
\begin{theorem}
  \label{thm:hoffman}
  Given a directed flow network $D = (V,A)$, with each arc having a
  lower bound $\ell_a$ and an upper bound $u_a$ (and $0 \leq \ell_a \leq
  u_a$), there exists a circulation $f: A \to \Re_+$ satisfying $\ell_a
  \leq f(a) \leq u_a$ for all arcs $a$ if and only if
  $\ell(\partial^+(U)) \leq u(\partial^-(U))$ for all $U \subseteq V$.
  Moreover, if the $\ell$ and $u$ are integral, then the circulation $f$
  can be taken integral.
\end{theorem}

\begin{proof}[Proof of Theorem \ref{thm:augment}]
Set lower bounds $\ell : A \rightarrow \{0,1\}$ on the arcs by:
\[
\ell_a = \left\{
\begin{array}{rl}
1 & {\rm if~} a \in \A {\rm ~or~} a = ts \\
0 & {\rm otherwise}
\end{array}
\right.
\]
For now, we set an upper bound of 1 on arc $ts$ and leave all other arc
upper bounds at $\infty$. We compute the minimum cost circulation
satisfying these bounds (we will soon see why one must exist). Since the
bounds are integral and since $\A$ is weakly connected, this circulation
gives us a directed Eulerian graph. Furthermore, since $u_{ta} =
\ell_{ta} = 1$, the $ts$ arc must appear exactly once in this Eulerian
graph. Our final Hamiltonian $s$-$t$ path is obtained by following an
Eulerian circuit, removing the single $ts$ arc from this circuit to get
an Eulerian $s$-$t$ walk, and finally shortcutting this walk past
repeated nodes. The cost of this Hamiltonian path will be, by the
triangle inequality, at most the cost of the circulation minus the cost
of the $ts$ arc.

Finally, we need to bound the cost of the circulation (and also to prove
one exists). To that end, we will impose stronger upper bounds $u : A
\rightarrow \mathbb R_{\geq 0}$ as follows:
\[
u_a = \left\{
\begin{array}{rl}
1 & \qquad {\rm if~} a = ts \\
1 + (1+ \tau^{-1})\alpha x^*_a & \qquad {\rm if~} a \in \A \\
(1+\tau^{-1})\alpha x^*_a & \qquad {\rm otherwise}
\end{array}
\right.
\]
We use Hoffman's circulation theorem to show that a circulation $f$
exists satisfying these bounds $\ell$ and $u$ (The calculations appear
in the next paragraph.) Since $u$ is no longer integral, the circulation
$f$ might not be integral, but it does demonstrate that a circulation
exists where each arc $a \neq ts$ is assigned at most
$(1+\tau^{-1})\alpha x^*_a$ more flow in the circulation than the number
of times it appears in $\A$. Consequently, it shows that the minimum
cost circulation $g$ in the setting where we only had a non-trivial
upper bound of $1$ on the arc $ts$ can be no more expensive (since there
are fewer constraints), and that circulation $g$ can be chosen to be
integral. The cost of circulation $g$ is at most the cost of $f$, which
is at most
\[ \sum_{a \in A} c_a u_a = \sum_{a \in \A} c_a + (1+\tau^{-1})\alpha \sum_{a \in
  A} c_a x^*_a + c_{ts}. \] Subtracting the cost of the $ts$ arc (since
we drop it to get the Hamilton path),
we get that the final Hamiltonian path has cost at most
\[ c(\A) + (1+\tau^{-1}) \alpha \sum_{a \in
  A} c_a x^*_a. \]



One detail remains: we need to verify the conditions of
\lref[Theorem]{thm:hoffman} for the bounds $\ell$ and $u$.  Firstly, it is
clear by definition that $\ell_a \leq u_a$ for each arc $a$. Now we need
to check $\ell(\partial^+(U)) \leq u(\partial^-(U))$ for each cut
$U$. This is broken into four cases.

\begin{enumerate}
\item $U$ is a $\tau$-narrow $s$-$t$ cut. Then $\ell(\partial^+(U)) =
  1$, since $\A$ contains only one arc in $\partial^+(U)$. But $1 =
  u_{ts} \leq u(\partial^-(U))$.

\item $U$ is an $s$-$t$ cut, but not $\tau$-narrow. Then by the
  $\alpha$-thinness of $\A$ and \lref[Claim]{claim:lp},
  \[ \ell(\partial^+(U)) \leq \alpha x^*(\partial^+(U)) = \alpha
  x^*(\partial^-(U)) + \alpha. \] On the other hand,
  \[
  u(\partial^-(U)) \geq(1+\tau^{-1})\alpha x^*(\partial^-(U)) = \alpha
  x^*(\partial^-(U)) + \tau^{-1}\alpha x^*(\partial^-(U)) \geq \alpha
  x^*(\partial^-(U)) + \alpha
  \]
  where the last inequality used the fact that $x^*(\partial^-(U)) \geq \tau$.

\item $U$ is a $t$-$s$ cut. Then by the $\alpha$-thinness of $B$ and \lref[Claim]{claim:lp},
  \[ \ell(\partial^+(U)) \leq 1 + \alpha x^*(\partial^+(U)) = 1 + \alpha
  x^*(\partial^-(U)) - \alpha \leq \alpha x^*(\partial^-(U)), \]
  the last inequality using that $\alpha \geq 1$.  Moreover
  \[ u(\partial^-(U)) \geq (1+\tau^{-1})\alpha x^*(\partial^-(U)) \geq
  \alpha x^*(\partial^-(U)). \] Then $\ell(\partial^+(U)) \leq
  u(\partial^-(U))$.

\item $U$ does not separate $s$ from $t$. Then
  \[ \ell(\partial^+(U)) \leq \alpha x^*(\partial^+(U)) = \alpha
  x^*(\partial^-(U)) \leq (1+\tau^{-1})\alpha x^*(\partial^-(U)) \leq
  u(\partial^-(U)) \]
\end{enumerate}
\end{proof}

The proof of our main result, \lref[Theorem]{thm:main}, follows immediately from
\lref[Theorem]{thm:augment} and \lref[Lemma]{lem:findtree} and setting $\tau = 1/4$.
Furthermore, this proof also shows that there is a randomized polynomial time algorithm
that constructs a Hamiltonian $s-t$ path witnessing this integrality gap bound with probability at least 1/2.
%
%
%
%

\section{Guaranteeing $\alpha$-Thinness}
\label{sec:chernoff}
We prove \lref[Theorem]{thm:chernoff} in this section.  Recall that
$\alpha$-thin means the number of arcs chosen from $\partial^+(U)$
should not exceed $\alpha x^*(\partial^+(U))$ (so a directed version).
Let $\alpha := \left(2+\frac{1}{\tau}\right)\cdot\frac{24 \log
  n}{\log\log n}$ where the logarithm is the natural logarithm. Recall
that $\A$ is the set of arcs found with corresponding undirected
spanning tree $T$.  By the first property of the distribution
(preservation of marginals on singletons) we have for each $\emptyset
\subsetneq U \subsetneq V$ that $\E[|\partial_T(U)|] =
\kappa(z)(\partial(U))$.

We have negative correlation on subsets of items, so we can apply standard concentration bounds.
Specifically, we use the following inequality.
\begin{theorem}{\cite[Theorem 3.4]{PS97}}\label{thm:ps}
Let $X_1, \ldots, X_n$ be given 0-1 random variables with $X = \sum_i X_i$ and $\mu = \E[X]$ such that for all $I \subseteq [n]$,
$\Pr[\bigwedge_{i \in I} X_i = 1] \leq \prod_{i \in I} \Pr[X_i = 1]$. Then for any $\delta > 0$ we have
\[ \Pr[X > (1 + \delta) \cdot \mu] \leq \left(\frac{e^\delta}{(1 + \delta)^{1+\delta}}\right)^{\mu}. \]
\end{theorem}

For notational simplicity, let $z' := \kappa(z)$.
\lref[Theorem]{thm:ps} immediately shows
\[ \Pr[|\partial_T(U)| \geq (1+\delta) z'(\partial(U))] \leq \left(\frac{e^\delta}{(1+\delta)^{(1+\delta)}}\right)^{z'(\partial(U))}.\]
Let $\sigma := \frac{6 \log n}{\log\log n}$ (again using the natural logarithm) and use \lref[Theorem]{thm:ps} with $\delta = \sigma-1$.
For $n \geq 7$, the above expression is bounded (in a manner similar to \cite{AGMSS}) by
\begin{eqnarray*}
\left(\frac{e}{\sigma}\right)^{\sigma z'(\partial(U))} \leq e^{-z'(\partial(U)) 5 \log n} = n^{-5 z'(\partial(U))}.
\end{eqnarray*}
However, for any graph, there are at most $n^{2l}$ cuts whose capacity
is at most $l$ times the capacity of the minimum
cut~\cite{KS96}. 
Note that the minimum cut with capacities $z'$ is 1, so
there are at most $n^{2l}$ cuts of the undirected graph $H$ with capacity (under $z'$) at most $l$.
Another way to view this is that there are at most $n^{2(l+1)}$ cuts whose capacity is between $l$ and $l+1$.
For each such cut $U$, the previous analysis shows that
probability that $|\partial_T(U)| > (1+\delta)z'(\partial(U))$ is at most $n^{-5l}$.
Thus, by the union bound, the probability that $|\partial_T(U)| > (1+\delta)z'(\partial(U))$
for some $\emptyset \subsetneq U \subsetneq V$ is bounded by
\[
\sum_{i=1}^\infty n^{2(i+1)}\cdot n^{-5i} \leq \sum_{i=1}^\infty n^{-i} = \frac{1}{n-1}
\]

Since $|\partial_\A^+(U)| \leq |\partial_{T}(U)|$, then we have just seen that with probability at least $1-\frac{1}{n-1}$ that there is no $\emptyset \subsetneq U \subsetneq V$
with $|\partial_\A^+(U)| > \sigma \cdot z'(\partial(U))$. This is close to what we want, except we should bound $|\partial_\A^+(U)|$ against the $x^*$-capacity of $U$.
That is, we ultimately want to show $|\partial_\A^+(U)| \leq \alpha \cdot x^*(\partial^+(U))$ for every $\emptyset \subsetneq U \subsetneq V$. To do this, we consider two cases.

\begin{itemize}
\item If either $U$ or $V-U$ is a $\tau$-narrow $s-t$ cut. Then we ignore the above analysis and simply note that by the properties of $z$ guaranteed by \lref[Lemma]{lem:fix_sol}
either $|\partial_\A^+(U)| = 1$ (if $s \in U$) or $|\partial_\A^+(U)| = 0$ (if $t \in U$), both of which are bounded by $\alpha \cdot x^*(\partial^+(U))$.
\item Otherwise, either $U$ or $V-U$ is an $s-t$ cut that is not $\tau$-narrow, or $U$ does not separate $s$ from $t$. In any case, we have
$x^*(\partial^+(U)) + x^*(\partial^-(U)) \leq 2 x^*(\partial^+(U)) + 1$ (by Claim \ref{claim:lp}) and $x^*(\partial^+(U)) \geq \tau$.
%
Since $\tau \leq 1/4$, then $\frac{1}{1-3\tau} \leq 4$ so $z' \leq 4x^*$. So,
\begin{eqnarray*}
|\partial_\A^+(U)| & \leq & \sigma \cdot z'(\partial(U)) \\
& \leq & 4\sigma  \cdot (x^*(\partial^+(U)) + x^*(\partial^-(U))) \\
& \leq & 8\sigma \cdot x^*(\partial^+(U)) + 4\sigma \\
& \leq & 8\sigma \cdot x^*(\partial^+(U)) + \frac{4\sigma}{\tau} \cdot x^*(\partial^+(U)) \\
& = & \alpha \cdot x^*(\partial^+(U)).
\end{eqnarray*}
%
%
Summarizing, for $n \geq 7$ we have with probability at least $1-\frac{1}{n-1}$ that
\[ |\partial^+_\A(U)| \leq \alpha x^*(\partial^+(U)) = \Theta \left(\frac{\log n}{\log\log n}\right) x^*(\partial^+(U)).\]
That is, $\A$ is $\alpha$-thin with high probability.

\end{itemize}

\section{Improved Bounds from Thin Tree Conjectures}
\label{sec:thin}


In \lref[Section]{sec:patching}, we defined thinness of a set of directed arcs with respect to an LP solution. Here, we define it for undirected graphs
with respect to the original graph itself.
\begin{definition}
Let $G = (V,E)$ be an undirected graph. A spanning tree $T$ of $G$ is said to be $\alpha$-thin if for every cut $U$
we have $|\partial_{T}(U)| \leq \alpha \cdot |\partial(U)|$.
\end{definition}

The following conjecture was given by Goddyn~\cite{goddyn}.
\begin{conjecture}\label{conj:thin}
There is some constant $\gamma$ such for any $k \geq 1$, any undirected $k$-edge connected graph has a $\frac{\gamma}{k}$-thin spanning tree.
\end{conjecture}

Oveis-Gharan and Saberi~\cite{GS11} show that assuming \lref[Conjecture]{conj:thin} is true, there is an $O(1)$-approximation for the ATSP problem by bounding the integrality gap of the subtour elimination LP. We generalize the result for ATSPP in \lref[Theorem]{thm:thin}. While the proof follows the same outline,  there are some technicalities that must be overcome in the case of ATSPP which we outline.

\begin{theorem}\label{thm:thin}
If \lref[Conjecture]{conj:thin} is true, then the integrality gap of the subtour elimination LP for ATSPP is at most $248\gamma + 60$.
\end{theorem}

\lref[Theorem]{thm:thin} follows immediately from \lref[Theorem]{thm:augment} once we show the following. For notational simplicity, we will set the value of $\tau$ to 1/4 for the remainder of this section.

\begin{lemma}\label{lem:thin}
If \lref[Conjecture]{conj:thin} is true, then we can find a $(48\gamma + 12)$-thin collection of arcs $\A$ of cost at most $8\gamma \cdot c(x^*)$ satisfying the
requirements of \lref[Theorem]{thm:augment}.
\end{lemma}

First we recall a result by Oveis Gharan and Saberi~\cite{GS11} which will play an important role in our proof. We state a more specific
form of their proposition.
\begin{proposition}{\cite{GS11}}\label{prop:thintree}
If \lref[Conjecture]{conj:thin} is true, then every $k$-edge connected graph $G(V,E)$ with edge costs $c_e \geq 0, e \in E$ has a
$\frac{2\gamma}{k}$-thin spanning tree with cost at most $\frac{2\gamma}{k} c(E)$.
\end{proposition}

\begin{proof}[Proof of Lemma \ref{lem:thin}]
Let $x^*$ be an optimum LP solution. We cannot invoke \lref[Proposition]{prop:thintree} directly on a scaled version of $\kappa(x^*)$
(as was done for ATSP in \cite{GS11}) since the resulting thin tree might cross the narrow cuts more than once or, perhaps, in the wrong direction.
Our solution will be to sample a thin tree on the subgraphs between narrow cuts and chain these together using arcs that cross the narrow cuts
to ensure the resulting tree crosses the narrow cuts exactly once.

Recall the definition of $\tau$-narrow cuts (again, we use $\tau = 1/4$ here)
and let $L_1, L_2, \ldots, L_{k+1}$ be the sets defined in \lref[Section]{sec:proof-of-lem}.
For every $1 \leq i \leq k+1$, let $x^i$ denote the restriction of $x^*$ to $L_i$. That is, $x^i_a = x^*_a$ if $a \in A(L_i)$ and $x^i_a = 0$ otherwise.
Recall by \lref[Corollary]{cor:parts} that $x^i(\partial(U; L_i-U)) \geq 1-2\tau = 1/2$ for any $\emptyset \subsetneq U \subsetneq L_i$.

For each $1 \leq i \leq k+1$ with $|L_i| \geq 2$, create an undirected graph $G_i(L_i, E_i)$ where $E_i$ will contain many copies of edges between nodes in $L_i$.
Similar to the proof of Theorem 5.3 in~\cite{GS11}, round down each $x^i_a$ value to its nearest multiple of $1/4n^3$ and call this value $z^i_a$.
Add $4n^3 \cdot z^i_a$ copies of the undirected version of arc $a$ to $E_i$ for each $a \in A(L_i)$, each with cost $c_a$.
Since $z^i_a\geq x^i_a-\frac{1}{4n^3}$, for every cut $U$ of $G_i$ we have $\kappa(z^i)(\partial(U)) \geq \kappa(x^i)(\partial(U)) - n^2/4n^3 \geq 1/2 - 1/(4n) \geq 1/4$.
Therefore, we have $\partial_{E_i}(U) \geq n^3$ for every cut $U$ of $G_i$.

By \lref[Proposition]{prop:thintree}, we may find a $\frac{2\gamma}{n^3}$-thin spanning tree $T_i$ of $G_i$ with cost bounded by
\[ \frac{2\gamma}{n^3} \cdot c(E_i) \leq \frac{2\gamma}{n^3} 4 n^3 c(x^i) = 8\gamma \cdot c(x^i). \]
Let $ \A_i$ be the original (directed) arcs of the graph $G$ that are used by $T_i$.

Next, for each $1 \leq i \leq k$, let $a_i$ denote the cheapest arc in $\partial(L_i;L_{i+1})$. By \lref[Lemma]{lem:li+}
with $\tau = 1/4$, $c_{a_i} \leq 4 \sum_{a \in \partial(L_i;L_{i+1})} c_a x^*_a$.

Finally, let $ \A = \left(\cup_{i=2}^k  \A_i\right) \cup \{a_i : 1 \leq i \leq k\}$ and note that because the cost of $ \A_i$ was
charged to the LP cost for arcs in $A(L_i)$ and the cost of $a_i$ was charged to the LP cost for edges in $\partial(L_i; L_{i+1})$, then
$c(\A) \leq \max\{8\gamma, 4\} c(x^*) = 8\gamma \cdot c(x^*)$ (clearly \lref[Conjecture]{conj:thin} can only hold for $\gamma \geq 1$).

From construction, $|\partial^+_{\A}(U)| = 1$ and $|\partial^-_{\A}(U)| = 0$ for any $\tau$-narrow cut $U$.
Since $\A$ is formed by chaining together weakly connected subgraphs in each $L_i$ using edges in $\partial(L_i;L_{i+1})$, it is weakly connected.

We finish by showing that $\A$ is $O(1)$-thin with respect to $x^*$. Consider any cut $U$ of $G$.
If $U$ or $V-U$ is a $\tau$-narrow cut then $|\partial_\A^+(\partial(U))| \leq x^*(\partial^+(U))$ by construction of $\A$ and feasibility of $x^*$ as a solution to the subtour elimination LP.

Otherwise, let $Q = \{a_i : 1 \leq i \leq k\} \cap \partial_\A^+(U)$ and let $Q_i = \partial^+_\A(U \cap L_i ; L_i - U) = \partial^+_{\A_i}(U)$ for each $1 \leq i \leq k+1$ with $|L_i| \geq 2$.
Note that $\partial_\A^+(U) = Q \cup \left(\bigcup_{i : |L_i| \geq 2} Q_i\right)$.

For each $2 \leq i \leq k$ we have
\begin{eqnarray*}
|Q_i| & = & |\partial_{\A_i}(U \cap L_i;L_i-U)| \\
& \leq & \frac{2\gamma}{n^3} \cdot |\partial_{E_i}(U \cap L_i;L_i-U))| \\
& \leq & \frac{2\gamma}{n^3} \cdot 4n^3 \kappa(x^*)(\partial(U \cap L_i;L_i-U)) \\
& = & 8 \gamma \cdot \kappa(x^*)(\partial(U \cap L_i;L_i-U))
\end{eqnarray*}

Finally, we bound $|Q|$. If $a_i \in Q$ then it cannot be the case that $L_i \cap U = \emptyset$ nor can it be the case that $L_{i+1} \subseteq U$.
So, at least one of the three following cases must hold:
\begin{enumerate}
\item $L_i-U \neq \emptyset$; we charge the occurrence of $a_i \in Q$ to the quantity $\kappa(x^*)(\partial(L_i \cap U; L_i - U)) \geq 1-2\tau = 1/2$
(cf.\ \lref[Corollary]{cor:parts}).
\item $L_{i+1} \cap U \neq \emptyset$; we charge the occurrence of $a_i \in Q$ to the quantity $\kappa(x^*)(\partial(L_{i+1} \cap U; L_{i+1} - U)) \geq 1/2$.
\item $L_i\subseteq U$ and $L_{i+1}\cap U=\emptyset$ and therefore, $\partial(L_i;L_{i+1}) \subseteq \partial^+(U)$. In this case, we charge the occurrence of $a_i \in Q$ to the quantity $x^*(\partial(L_i; L_{i+1})) \geq 1-3\tau \geq 1/2$
(cf.\ \lref[Lemma]{lem:li+}).
\end{enumerate}
In each of the cases, the edges whose $x^*$ values were charged all lie in $\partial^+(U)$ or $\partial^-(U)$. Furthermore, every edge is charged at most twice
this way. If $e \in \partial(L_i; L_{i+1})$ then it is charged at most once (for $a_i$), if $e \in A(L_i)$ then it is charged at most once for $a_{i-1}$ and at most
once for $a_i$. Overall, we see $|Q| \leq 2\kappa(x^*)(\partial(U))$.

Considering all of these bounds, we have
\begin{eqnarray*}
|\partial_{\A}^{+}(U)| & = & |Q| + \sum_{i=2}^k |Q_i| \\
& \leq & 2 \cdot \kappa(x^*)(\partial(U)) + 8 \gamma \sum_{i=2}^k \kappa(x^*)(\partial(U \cap L_i;L_i-U)) \\
& \leq & (8\gamma + 2) \kappa(x^*)(\partial(U)) \\
& = & (8\gamma + 2) \cdot (x^*(\partial^+(U)) + x^*(\partial^-(U))) \\
& \leq & (8\gamma + 2) \cdot \left(x^*(\partial^+(U)) + \left(\frac{1}{\tau} + 1\right)x^*(\partial^+(U))\right) \\
& = & (8\gamma + 2) \cdot \left(\frac{1}{\tau} + 2\right) \cdot x^*(\partial^+(U)) \\
& = & (48 \gamma + 12) \cdot x^*(\partial^+(U))
\end{eqnarray*}
\end{proof}

The collection of arcs $\A$ is $(48 \gamma + 12)$-thin and has cost at most $8\gamma \cdot c(x^*)$.
Furthermore, $\A$ satisfies $|\partial_\A^+(U)| = 1$ and $|\partial_\A^-(U)| = 0$ for every  $\tau$-narrow $s-t$ cut $U$.
From \lref[Theorem]{thm:augment}, we can then obtain a ATSPP solution with cost at most
\[ c(\A) + (1 + \tau^{-1})(48\gamma + 12) c(x^*) \leq (248 \gamma + 60) \cdot c(x^*). \]
This completes the proof of \lref[Theorem]{thm:thin}.

We have not attempted to optimize the constants in our analysis. For example, a more careful scaling of $x^*$ to get the $z_a$ values
in the above proof will improve the constants.

\section{A Simple Integrality Gap Example}
\label{sec:int-gap}

In this section, we show that the integrality gap of the subtour
elimination LP~(\ref{eq:lp}) is at least~$2$. This result can also be
inferred from the integrality gap of $2$ for the ATSP tour
problem~\cite{CGK06}, but our construction is relatively simpler.

For a fixed integer $r \geq 1$, consider the directed graph $G_r$
defined below (and illustrated in \lref[Figure]{fig:gap}). The vertices of
$G_r$ are $\{s,t\} \cup \{u_1, \ldots, u_r\} \cup \{v_1, \ldots, v_r\}$;
the arcs are as follows:
\begin{OneLiners}
\item $\{ su_1, sv_1, u_rt, v_rt \}$, each with cost 1,
\item $\{ u_1v_r, v_1u_r \}$, each with cost 0,
\item $\{u_{i+1}u_i \mid 1 \leq i < r\} \cup \{v_{i+1}v_i \mid 1 \leq i < r\}$,
  each with cost 1,
\item and $\{u_iu_{i+1} \mid 1 \leq i < r\} \cup \{v_iv_{i+1} \mid 1 \leq i <
  r\}$, each with cost 0.
\end{OneLiners}
Let $F_r$ denote the ATSPP instance obtained from the metric completion
of $G_r$.

\begin{figure}
 \centering
    \scalebox{0.65}{\psfig{figure=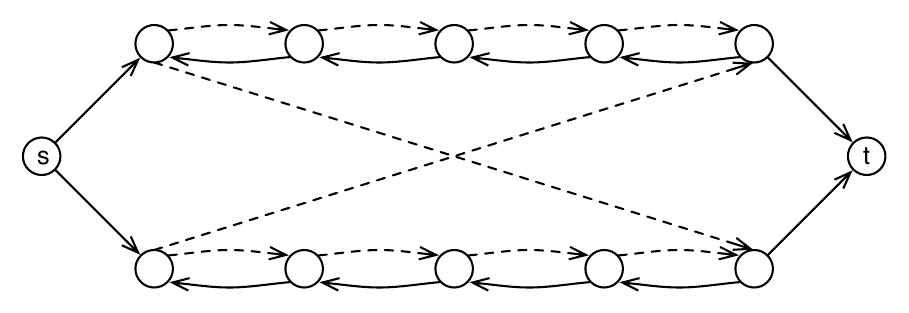}}
  \caption{The graph $G_r$ with $r = 5$. The solid arcs have cost 1 and
    the dashed arcs have cost 0.}\label{fig:gap}
\end{figure}

\begin{lemma} \label{lem:intgap}
  The integrality gap of the LP~\ref{eq:lp} on the instance $F_r$ is
  at least $2-o(1)$.
\end{lemma}

\begin{proof}
  It is easy to verify that assigning $x_a = 1/2$ to each arc that originally
  appeared in $G_r$ is a valid LP solution. Indeed, the degree constraints are immediate,
  and there are two edge-disjoint paths from $s$ to every other node in
  $G_r$ (so there must be at least 2 arcs exiting any subset containing $s$)
  so the cut constraints are also satisfied. The total cost of
  this LP solution is $r+1$.

  On the other hand, we claim that the cost of any Hamiltonian $s$-$t$
  path in $F_r$, which corresponds to a spanning $s$-$t$ walk $W$ in
  $G_r$, is at least $2r-1$. This shows an integrality gap of
  $\frac{2r-1}{r+1} = 2 - o(1)$.

  To lower-bound the length of any spanning $s$-$t$ walk, we first argue
  that the walk $W$ can avoid using at most one of the unit cost arcs
  of the form $u_{i+1}u_i$ or $v_{i+1}v_i$. Indeed, any $u_r$-$v_r$ walk
  must use arcs $u_{i+1}u_i$ for every $1 \leq i < r$. Similarly, every
  $v_r$-$u_r$ walk must use all arcs of the form $v_{i+1}v_i$. One of
  $u_r$ and $v_r$ is visited before the other, so either all of the
  $u_{i+1}u_i$ arcs or all of the $v_{i+1}v_i$ arcs are used by
  $W$. Now suppose, without loss of generality, that $W$ does not use
  the arcs $u_{i+1}u_i$ and $u_{j+1}u_j$ for $1 \leq i < j < r$.  Every
  $u_{i+1}$-$v_r$ walk uses arc $u_{i+1}u_i$ and every $v_r-u_{i+1}$
  walk uses arc $u_{j+1}u_j$. Since one of $u_{i+1}$ or $v_r$ must be
  visited by $W$ before the other, then $W$ cannot avoid both
  $u_{i+1}u_i$ and $u_{j+1}u_j$ which contradicts our assumption.

  Thus, $W$ must use all but at most one of the $2r-2$ unit cost arcs
  in $\{u_{i+1}u_i \mid 1 \leq i < r\} \cup \{v_{i+1}v_i \mid 1 \leq i <
  r\}$. Moreover, $W$ must also use one of the arcs exiting $s$ and one
  of the arcs entering $t$, so the cost of $W$ is at least $2r-1$.
  (In fact, the walk
\[ \langle s, u_1, v_r, v_{r-1}, \ldots, v_1, u_r, u_{r-1}, \ldots, u_3, u_2, u_3, \ldots, u_r, t \rangle \]
 is of length exactly $2r-1$, so this argument is tight.)
\end{proof}

\section{Conclusion}


In this paper we showed that the integrality gap for ATSPP is $\factor$.
We also show that a constant integrality gap bound follows
from the form of Goddyn's conjecture used in \cite{GS11}
to get an analogous ATSP integrality gap bound.
We also showed a simpler construction achieving a lower
bound of $2$ for the subtour elimination LP.
One of the main open questions following this work is to
show a more general reduction: does an $\alpha$ integrality
gap bound for ATSP directly imply an $O(\alpha)$ integrality gap bound
for ATSPP without any further assumptions?

\subsubsection*{Acknowledgments.} We thank V.~Nagarajan for enlightening
discussions in the early stages of this project. Z.F.\ and A.G.\ also
thank A.~Vetta and M.~Singh for their generous hospitality. Part of this work was done when
Z.F.\ was a postdoctoral fellow in the Department of Combinatorics and Optimization at the University of
Waterloo, when A.G.\ was visiting the IEOR Department at Columbia
University, and  when M.S.\ was at McGill University.
Finally, we thank anonymous reviewers for many helpful comments and the suggestion to obtain better bounds through thin tree conjectures.

\bibliographystyle{abbrv}
{\small \bibliography{atsp}}

\end{document}